\documentclass[11pt]{article}
\usepackage[english]{babel}
\usepackage{amssymb}
\usepackage{amsthm}
\usepackage{amsmath}
\usepackage{a4wide}
\usepackage{esvect}
\usepackage{hyperref}
\usepackage{verbatim}
\usepackage[dvips]{graphicx}
\usepackage{color}
\usepackage{geometry}
\usepackage{framed}
\newtheorem{thm}{Theorem}
\newtheorem{lem}[thm]{Lemma}
\newtheorem{lemma}[thm]{Lemma}

\newtheorem{prop}[thm]{Proposition}

\newtheorem{defn}[thm]{Definition}

\newtheorem{fact}[thm]{Fact}
\newtheorem{cor}[thm]{Corollary}

\newtheorem{clm}[thm]{Claim}

\newcommand\cD{{\mathcal D}}

\newcommand\cF{{\mathcal F}}
\newcommand\cG{{\mathcal G}}
\newcommand\cH{{\mathcal H}}

\newcommand{\ignore}[1]{}

\title{Conscious and controlling elements in combinatorial 
group testing problems with more defectives}

\begin{document}

\author{
D\'aniel Gerbner\thanks{Research supported by the J\'anos Bolyai Research Fellowship of the Hungarian Academy of Sciences.} \thanks{Research supported by the National Research,
Development and Innovation Office -- NKFIH, grant K116769.}
\and
M\'at\'e Vizer\thanks{Research supported by the National Research, Development and Innovation
Office -- NKFIH, grant SNN 116095.}}

\date{MTA R\'enyi Institute \\
Hungary H-1053, Budapest, Re\'altanoda utca 13-15.\\
\small \texttt{gerbner@renyi.hu, vizermate@gmail.com}\\
\today}

\maketitle

\begin{abstract}
In combinatorial group testing problems Questioner needs to find a defective element $x \in [n]$ by testing subsets of $[n]$. In \cite{GV2017} the authors introduced a new model, where each element knows the answer for those queries that contain it and each element should be able to identify the defective one.

In this article we continue to investigate this kind of models with more defective elements. We also consider related models inspired by secret sharing models, where the elements should share information among them to find out the defectives. Finally the adaptive versions of the different models are also investigated.

\end{abstract}

\vspace{4mm}

\noindent
{\bf Keywords:} Combinatorial Group Testing, defectives, cancellative 

\noindent
{\bf AMS Subj.\ Class.\ (2010)}:  	94A50

\section{Introduction}
In the most basic \textit{model} of \textit{combinatorial group testing} Questioner needs to find a special element $x$ of $\{1,2,...,n\}(=:[n])$ by asking minimal number of \textit{queries} (or \textit{group tests} or \textit{pools}) of type "does $x \in F \subset [n]?$". Special elements are usually called \textit{defective} (or \textit{positive}). For every combinatorial group testing problem there are at least two main approaches: whether it is \textit{adaptive (or sequential)} or \textit{non-adaptive(or oblivious)}. In the adaptive scenario Questioner asks queries depending on the answers for the previously asked queries, however in the non-adaptive version Questioner needs to pose all the queries at the beginning. We call the \textit{complexity} of a specific combinatorial group testing problem the number of the queries needed to ask by Questioner in the worst case during an optimal strategy. 


Combinatorial group testing problems were first considered during the World War II by Dorfman \cite{D1943} in the context of mass blood testing. Since then group testing techniques have had many different applications, for example in fault diagnosis in optical networks  \cite{HPWYC2007}, in quality control in product testing
\cite{SG1959} or failure detection in wireless sensor networks \cite{LLLG2013}. In this article we will mainly discuss non-adaptive models. The interested reader can find many variants and generalizations of the basic non-adaptive model and also many applications in the book \cite{DH2006}.


\subsection*{Description of the new model}

In \cite{GV2017} the authors introduced new combinatorial group testing models, inspired by the results of Tapolcai et al.~\cite{TRHHS2014,TRHGHS2016}.

The main novel ingredient of these combinatorial group testing models is that the elements are conscious and they distrust the Questioner, thus they want to control the tests they are involved in. So we introduced the following extra condition: {\bf each element knows the answer for those queries that contain it}, and the goal: each element should be able to identify the defective one.

Motivated by secret sharing schemes (see e.g. \cite{B2011}), we also consider the following variant: the elements can work together and share their knowledge. In this case we require certain sets of elements to be able to identify the defective, while we require other sets to be unable to identify the defective element.
We emphasize that we do not deal with the way the data is transmitted. Information can not be distributed between different groups.

\vspace{2mm}

We mention here some other motivation to introduce these models: it is often mentioned in the group testing literature that an advantage of testing pools together is that it increases privacy. However, systematical research on this property has only started recently, see e.g. \cite{AFBC2008,CCG2016,EGH2013}. These papers focus on cryptographic versions of the problem. Here we deal with a simple combinatorial version, where privacy only means that an unauthorized participant cannot completely detect the defective element(s). In \cite{GV2017} the authors considered models with one defective element. The main aim of this article is to continue these investigations with more defectives.

\subsection*{Simple combinatorial models with $d$ defectives}

About Questioner's strategy we remark that - as he should find all the defectives - the asked queries should form a \textit{d-separating family} (see the next section for a definition) in the non-adaptive case, so for the minimum number of tests the known lower bound is $\Omega(\frac{d^{2}}{\log d} \log n)$, while the best upper bound construction yields $O(d^{2} \log n)$. It is one of the major open problems in the theory of combinatorial group testing models to close the gap between the previous upper and lower bound.

In the adaptive case there is a multiplicative constant factor between the information theoretic lower bound and the best existing algorithm. The known best lower bound is $d \log \frac{n}{d}$, while the upper bound is $O(d \log n)$.





\subsection*{Structure of the paper} We organize the paper as follows: in Section 2 we introduce some properties and related results about families of sets, that we will need later. In Section 3 we introduce the non-adaptive models that we investigate, while in Section 4 we prove the main results. In Section 5 we look at the adaptive scenario, and we finish the article with remarks and open questions in Section 6.

\vspace{3mm}

We also mention that in this article we use standard asymptotic notation.

\section{Finite set theory background}

Our topic is connected to several areas of finite set theory. In this section we introduce some notions on families of subsets and known results about them, that we will use during the proofs.

In this article we use the notation of $2^{[n]}$ for the power set of $[n]$ and for any $\mathcal{F} \subset 2^{[n]}$, $a\in [n]$ we use  $\cF_a:=\{F\in \cF: a\in F\}$. The \textit{complement} of a family $\cF\subset 2^{n}$ is $\overline{\cF}:=\{[n] \setminus F: F \in \cF\}$, while the \textit{dual} of a family $\cF\subset 2^{n}$ is $\cF':=\{\cF_a: a\in [n]\}$. It is defined on the underlying set $\cF$ and has cardinality at most $n$. For a family $\mathcal{F} \subset 2^{[n]}$ and $d \ge 1$ let $\mathcal{F}^d:=\{\cup_{i=1}^dF_i: F_i\in \mathcal{F}, \ F_i\neq F_j\, \ \textrm{for}\, i\neq j\}$.

\vspace{2mm}

Now we introduce some notions about families of subsets of $[n]$.

\begin{defn} We say that $\cF \subset 2^{[n]}$ is: 

\vspace{2mm}

$\bullet_1$ \textbf{intersection closed} if $F,G\in \cF$ implies $F\cap G\in\cF$.

\vspace{2mm}

$\bullet_2$ \textbf{Sperner} if there are no two different $F_1,F_2 \in \mathcal{F}$ with $F_1 \subset F_2$.

\vspace{2mm}

$\bullet_3$ \textbf{cancellative} if for any three $F_1,F_2,F_3 \in \mathcal{F}$ we have $$F_1 \cup F_2 = F_1 \cup F_3 \Rightarrow F_2=F_3.$$

$\bullet_4$ \textbf{intersection cancellative} if for any three $F_1,F_2,F_3 \in \mathcal{F}$ 
we have $$F_1 \cap F_2 = F_1 \cap F_3 \Rightarrow F_2=F_3.$$

$\bullet_5$ \textbf{d-separating} for some $1 \le d \le n-1$ positive 
integer, if for any two different 

\hspace{4mm}$X_1,X_2 \subset [n]$ with $|X_1|=|X_2|=d$ there is $F \in \mathcal{F}$ with: $$ F \cap X_1 \neq \emptyset \textrm{ and } F \cap X_2 = \emptyset, \textrm{ or }$$ $$ F \cap X_2 \neq \emptyset \textrm{ and } F \cap X_1 = \emptyset.$$  


$\bullet_6$ \textbf{d-union-free} for some $d \ge 1$ if for different $F_1, \dots, F_d \in \cF$ and different

\hspace{5mm} $G_1, \dots, G_d \in \cF$ $$\bigcup^{d}_{i=1} F_i = \bigcup^{d}_{i=1}G_i$$

\hspace{4mm}implies $\{F_1, \dots, F_d \}=\{G_1, \dots, G_d\}$.

$\bullet_7$ \textbf{d-cover-free} for some $d \ge 1$ positive integer if there are no $(d+1)$ 

\hspace{5mm}different $F_1,F_2,...,F_{d+1} \in \mathcal{F}$ with $$F_{d+1} \subset \bigcup_{i=1}^{d}F_i.$$

$\bullet_8$ \textbf{(r,d)-cover-free} for some $rd \ge 1$ positive integers if there are no $(d+r)$ 
different 

\hspace{4mm}$F_1,F_2,...,F_{d+r} \in \mathcal{F}$ with $$\bigcap_{i=d+1}^{d+r}F_{d+1} \subset \bigcup_{i=1}^{d}F_i.$$

\end{defn}

Before defining the last notion, we need some introduction. We will generalize a graph property, so it is more comfortable to use the word hypergraph instead of family of subsets of $[n]$ (where $\cF$ is the set of the hyperedges and $[n]$ is the set of vertices of the hypergraph). There are several ways to define cycles in hypergraphs. Here we use one due to Berge \cite{B1989}. A \textit{Berge-cycle} in a hypergraph of length $k$ (a Berge-$C_k$) consists of $k$ different hyperedges $E_1,\dots, E_k$ and $k$ different vertices $x_1,\dots x_k$ such that $E_i$ contains $x_i$ and $x_{i+1}$ for $1 \le i\le k$ (modulo $k$, so $E_k$ contains $x_k$ and $x_1$). Note that for a $2$-uniform hypergraph (that is a graph) this notion is the same as the 'usual' cycle in a graph. The \textit{Berge-girth} (that we call just \textit{girth} in this article) of a hypergraph $\cH$ is the smallest length of a cycle in $\cH$ (that is $\infty$ if there is no cycle in $\cH$). A hypergraph is $d$-regular if every vertex is contained in exactly $d$ hyperedges, $r$-uniform if every hyperedge has size $r$ and linear if any two hyperedges intersect in at most one vertex.

\subsubsection*{Some known results about these notions that we will use later}

\noindent $\bullet$ The notion cancellative was first introduced by Frankl and F\"uredi in \cite{FF84}.







\begin{fact}\label{intcan}

$\mathcal{F} \subset 2^{[n]}$ is intersection cancellative if and only if $\overline{\mathcal{F}}$ is cancellative.

\end{fact}

\vspace{2mm}

\noindent $\bullet$ The notion of separating family in the context of combinatorial search theory was introduced and first studied by R\'enyi in \cite{R1961}. The following fact is rather trivial, so we omit its proof.

\begin{fact}

Suppose $\mathcal{F}_n \subset 2^{[n]}$ is a minimal separating family. Then we have:
$$|\mathcal{F}_n| \le \lceil \log_2 n \rceil.$$

\end{fact}

\begin{fact}\label{dsepdual} $\mathcal{F} \subset 2^{[n]}$ finds $d$ defectives if and only if $\mathcal{F}$ is $d$-separating. The dual of a $d$-separating family is $d$-union-free.

\end{fact}




\vspace{3mm}

\noindent
$\bullet$ The notion of $d$-union-free families was introduced by Hwang and T. S\'os in \cite{HS1987} under the name of \textit{$d$-Sidon} families. They proved the following:

\begin{thm}\label{dunionfree} (Hwang, T. S\'os, \cite{HS1987}, Theorem 3) There exists a $d$-union-free family \newline $\cF_n \subset 2^{[n]}$ with:  $$\frac{1}{2}(1+\frac{1}{(4d)^{2}})^{n} \le |\cF_n|.$$ 

\end{thm}

\vspace{3mm}
\noindent
$\bullet$ The notion of $d$-cover-free families was introduced by Kautz and Singleton in \cite{KS1964}. Note that a $d$-cover-free family is also $d$-union-free. They proved the following lower bound.

\begin{thm}\label{rcoverfreelow} (Kautz, Singleton, \cite{KS1964})
There exists  an $d$-cover-free family $\mathcal{F}_n \subset 2^{[n]}$ with:
$$ \Omega(\frac{1}{d^2})=\frac{ \log_2 |\mathcal{F}_n|}{n}.$$

\end{thm}

\noindent
D'yachkov and Rykov proved the following upper bound on the size of $d$-cover-free families:

\begin{thm}\label{rcoverfree} (D'yachkov, Rykov, \cite{DR1982})
Suppose that $\mathcal{F}_n \subset 2^{[n]}$ is an $d$-cover-free family. Then we have:
$$\frac{ \log_2 |\mathcal{F}_n|}{n} \le \frac{2 \log_2 d}{d^2} (1 + o(1)).$$

\end{thm}

\vspace{3mm}

\noindent
$\bullet$ The notion of $(r,d)$-cover-free families were introduced by 
D'yachkov, Macula, Torney and Vilenkin in \cite{DVTM2002}. They showed that a result by Stinson, Wei and Zhu \cite{SWZ2000} implies the following:

\begin{thm}\label{rdcoverfree}
If the dual of $\mathcal{F}_n \subset 2^{[n]}$ is $(r,d)$-cover-free, then we have: $$|\mathcal{F}_n|=\Omega_{d,r}(\log_2n).$$

\end{thm}

\vspace{3mm}

\noindent $\bullet$ Ellis and Linial ~\cite{EL14} studied regular uniform linear hypergraphs with large girth. They mention that a result of Cooper, Frieze, Molloy and Reed \cite{CFMR96} implies that for any $d\ge 2$, $r,g\ge 3$ and sufficiently large $n$, if $r$ divides $n$, then there is an $r$-uniform, $d$-regular, $n$-vertex linear hypergraph with girth at least $g$. Moreover the argument can be adapted to show the same statement in the case $r$ divides $dn$.

\begin{thm}\label{girth} Let  $d\ge 2$, $r,g\ge 3$ and $n$ large enough such that $r$ divides $dn$. Then there exists a linear, $d$-regular, $r$-uniform hypergraph with girth at least $g$ on $n$ vertices.

\end{thm}

\section{Models}



In this section we start our investigations and give a systematic study of models with the extra property that each element knows the answers for those queries that contain it.

In all the models in this section an input set $[n]$ is given, and $d$ of them are defectives ($d \le n$). We are dealing with non-adaptive models, so Questioner needs to construct a family $\cF \subset 2^{[n]}$. A set $F$ correspond to a query of the following type: 'is $x \in F \subset [n]$?'. In each model we assume that knowing all the answers is enough information for Questioner to find the defective elements, i.e. $\cF$ is $d$-separating. Note that this immediately implies a lower bound of $\Omega_d(\log_2 n)$ on the size of the query family in each model. We mention whenever the query family satisfies another property that could improve the factor depending only on $d$, but calculating the factors is outside the scope of this paper.

The main difference between the following models is what we want the elements to find out. Using only the information available to them, i.e. the answers to the queries containing them, we can require that they find out something about the defective elements, or oppositely, that they cannot find out something. 

When we say that an element $x$ \textit{knows the defective elements}, we mean that the query family satisfies the following property: no matter what the defective element is, after the answers $x$ can find out the defective ones, i.e. the subfamily $\cF_x$ is $d$-separating. In the opposite when we say that $x$ \textit{does not know any of the defective elements}, we mean that the query family satisfies the following property: no matter what the defective elements are, after the answers $x$ cannot identify any defective element. Equivalently, for any $D\subset [n]$, $y \in D$ with $|D|=d$ there is a $D'\subset [n]$ with $|D'|=d$, $y\not\in D'$, such that the same members of $\cF_x$ intersect $D$ and $D'$.


Another variant of this problem is when elements can share information among them. It is possible that in some model some element can not find out the defective, however if we pick two elements and they share their information among them, they can find the defective elements. We consider these kind of models.

We also assume that in each model the elements know the setup of the problem, i.e. that $n$ elements are given and exactly $d$ of them are defectives. We use the expression that a family \textit{solves} a model if it satisfies the property that describes the model.
\vspace{3mm}

In each of the following models we first give a property describing  what the elements should know, and then we examine if there is a query set that solves that specific model or state results about the cardinality of such query sets. Then we consider models where we require some information to remain hidden from the elements. Finally we mix these types of properties.

In this section we assume that there are exactly $d\ge 2$ defective elements (and every element knows that). We consider models analogous to the ones introduced in \cite{GV2017}.

\subsection{Model $1_d$}

Probably the most natural model is the following:


\vspace{2mm}

\textbf{Property:} all elements find out if they are defective.

\vspace{2mm}

\noindent
We note that some cryptographic problems concerning this model were investigated in \cite{AFBC2008}, where the authors observed that the dual of a $d$-cover-free family solves this model. Here we show that only such families solve this model.

\begin{thm}\label{model1d} $\cF$ solves Model $1_d$ if and only if its dual is $d$-cover-free.

\end{thm}

\noindent
By Theorem \ref{model1d}, Theorem \ref{rcoverfreelow} and Theorem \ref{rcoverfree} we have:

\begin{cor} If $\cF_n \subset 2^{[n]}$ solves Model $1_d$ and has minimum cardinality, then we have:
$$\Omega(\frac{d^{2}}{\log_2 d}\log_2 n)=|\cF_n| = O(d^{2}\log_2 n ).$$

\end{cor}

\subsection{Model $2_d$}

Another natural model is when the elements should find out everything.

\vspace{2mm}

\textbf{Property:} every element finds all the defectives.

\vspace{2mm}

\noindent
It is obvious that no $\mathcal{F}$ can solve Model $2_d$ if $1<d<n$: a defective element cannot gather any information about the other elements, as it gets only YES answers.

\subsection{Model $2'_d$}

As defective elements cannot gather any information about the other elements, in the next model we only require non-defective elements to find the defective ones.

\vspace{2mm}

\textbf{Property:} every non-defective element finds all the defectives.

\vspace{1mm}

\begin{thm}\label{model2dprime} Suppose $\mathcal{F}_n$ solves Model $2_d'$ and has minimum cardinality. Then we have $$ |\mathcal{F}_n|= \Theta_d (\log_2 n).$$

\end{thm}


\begin{proof}

We claim that the solution is the dual of a $d$-cover-free family, and the dual of a $(2,d)$-cover-free family is always a solution. This together with Theorem \ref{rcoverfree} and Theorem \ref{rcoverfree} implies the statement.

Suppose that the dual is not $d$-cover-free. Then there are $F_1,F_2,...,F_{d+1} \in \mathcal{F}$ with $F_{d+1} \subset \cup_{i=1}^{d}F_i.$ For the corresponding elements in the primal version we have $x_{d+1}$ such that any set $F\in \mathcal{F}$ contains one of the other elements $x_1,\dots, x_d$. Thus if $x_1,\dots, x_d$ are the defectives, $x_{d+1}$ teceives only YES answers, thus cannot distinguish this case from the case $x_{d+1}$ and any $d-1$ other elements are the defectives.

On the other hand let us assume $\mathcal{F}$ is the dual of a $(2,d)$-cover-free family. Then for every non-defective elements $x,y$ there is a set $F\in\mathcal{F}$ such that $F$ contains both $x$ and $y$, but none of the defective elements, thus $x$ finds out that $y$ is not defective, $\cF$ solves Model $2_d'$. Indeed, there is an element in the intersection of the duals of $x$ and $y$ that is not contained by the duals of the defective elements by the $(2,d)$-cover-free property. That element is the dual of a set $F\in\mathcal{F}$ that has the desired properties.


\end{proof}

\subsection{Model $2''_d$}

The fact that defective elements cannot gather any information about the other elements shows that even $d-1$ elements together cannot always find the defectives. However, if $d$ elements share information, then either they are all the defectives and they do not need to gather information about the other elements, or at least one of them is not a defective, and then there is a solution by Model $2_d'$.

\vspace{2mm}

\textbf{Property:} $d$ elements together know who the defective elements are.

\vspace{2mm}

\begin{thm}\label{model2''d} $\mathcal{F}_n$ solves Model $2_d''$ if and only if its dual $\mathcal{G}$ is $d$-union-free and $\mathcal{G}^d$ is Sperner and intersection-cancellative.

\end{thm}

Note that we know the maximum possible size of a Sperner and intersection cancellative family (by  results of Frankl and F\"uredi \cite{FF84} and Tolhuizen \cite{T00}), but we do not know if that construction can be written as $\cG^d$ for a $d$-union-free family $\cG$.

\begin{thm} Suppose $\mathcal{F}_n$ solves Model $2_d''$ and has minimum cardinality. Then we have $$|\mathcal{F}_n|= \Theta_d( \log_2 n).$$

\end{thm}

\begin{proof} It is easy to see that if a family solves both Model $1_d$ and Model $2_d'$, then it also solves Model $2_d''$. As we have seen in the proof of Theorem \ref{model2dprime}, a solution for Model $2_d'$ is the dual of a $d$-cover-free family, thus it also solves Model $1_d$ by Theorem \ref{model1d}. This implies the upper bound.


\end{proof}


\subsection{Model $3_d$}

Let us now examine the case when we require that elements do \emph{not} find the defective.
Note that as always, we assume that knowing all the answers is enough to find the defective element.
\vspace{2mm}

\textbf{Property:} no element knows any of the defective ones.

\vspace{2mm}

\noindent
Note that for $d=1$ there is a solution for Model $2_d$ and there is no solution for Model $3_d$ \cite{GV2017}. For $d \ge 2$ the situation is just the opposite: we will show that there is a solution for Model $3_d$ for $n$ large enough. We will use arguments similar to the ones used in \cite{BGPV2015}.

\begin{thm}\label{model3d} If $d\ge 2$, $r\ge 3$ and $n\ge dr+2$, then an $r$-uniform, $d$-regular linear hypergraph with girth at least 5 solves Model $3_d$.

\end{thm}

\begin{proof} Let us consider an $r$-uniform, $d$-regular linear hypergraph $\cF$ of girth 5. For an arbitrary element $x$ its neighborhood consists of $d$ disjoint sets of size $r-1$. Also, there are more than $d$ elements not in its neighborhood. It is easy to see that by $r\ge 3$ $x$ cannot identify any defective elements.

On the other hand, if we know all the answers, the YES answers form stars with the defective elements in the centers. The elements that get only YES answers are the candidates for being defective. Every candidate that is not defective has to be connected to all the defectives. Two such candidate would form a Berge-$C_4$ with any two of the defective elements, thus there is only one additional candidate. But then it is the only one among the $d+1$ candidates that is connected to the other candidates, otherwise we could find a Berge-$C_3$.
\end{proof}

\begin{cor}\label{model3dcor} If $n$ is large enough compared to $d>1$, then there is a solution for Model $3_d$.

\end{cor}

\begin{proof}

If $d\ge 3$, let us choose $r=d$, then Theorem \ref{girth} shows that we can find such a family. If $d=2$, then Theorem \ref{girth} with $r=4$ shows we can find such a family for $n$ even. If $n$ is odd, we find such a family for $n+1$, and delete an element. The resulting family is not 4-uniform, but that property is not actually needed (in fact, we used only that every set in $\cF$ has size at least 3).

\end{proof}

\subsection{Model $4_d$}

Now we start to investigate models where elements can share information among them. Let $i$ and $j$ be integers with $1 \le i < j \le n$. When we say that a set of $j$ elements together know the defective elements, we mean that knowing the answers to all the queries containing at least one element from the set is enough to find all the defectives. Similarly, when we say that a set of $i$ elements do not know any of the defectives, we mean that knowing the answers to all the queries intersecting the set is not enough to identify any of the defective elements.
\vspace{2mm}

\textbf{Property:} any $j$ elements together know the defectives, but $i$ elements together do not know any of the defectives, for some $i$ and $j$ with $1 \le i < j \le n$.

\vspace{2mm}

Note that Corollary \ref{model3dcor} shows that there is a solution if $d>1$, $i=1$ and $j=n$, where  $n$ is large enough compared to $d$. In fact any $n-r+1$ elements together know the answer to all the queries, thus it is enough to assume $j\ge n-r+1$. A more precise version of Theorem \ref{girth} (see \cite{EL14}, Theorem 5) shows that about $\frac{n^{1/6}}{d}$ can be chosen as $r$, which shows that $j$ can be as small as $n-\frac{n^{1/6}}{d}$.

\begin{prop} If $i\ge d$ or $j<d$, then there is no solution.

\end{prop}

\begin{proof} Let us assume first we are given $j$ elements. If all of them are defectives, they only get YES answers, and do not gather any information about the other elements.

If $i\ge d$, for a set $X$ let $\mathcal{F}_X:= \cup_{x\in X}\mathcal{F}_x$. Let us consider among the $d$-element sets $X$ such that $\mathcal{F}_X$ is maximal. We claim that if the elements of X are the defectives, they can find it out by sharing information. Indeed, they get only YES answers. If they cannot be sure that they are the defective ones, then there is another $d$-set $Y$ that could be the set of defectives. It means all the answers to the queries in $\mathcal{F}_X$ would still be YES if $Y$ was the set of defectives, i.e. $\mathcal{F}_X \subseteq\mathcal{F}_Y$. By the assumption on $X$ we have $\mathcal{F}_X=\mathcal{F}_Y$, but then the family is not $d$-separating.

\end{proof}

\begin{prop} If $j=d$, then there is no solution.

\end{prop}

\begin{proof} If $x$ receives only YES answers, he cannot find out he is defective, thus there is a set $D=\{Y_1,\dots, y_d\}$ no containing $X$ that intersects every member of $\mathcal{F}_x$. On the other hand, if $x$ and $y_1,\dots, y_{d-1}$ are the defectives, they together can figure that out. In particular, they know that the set of defectives is not $D$, thus there is a set intersecting $\{x,y_1,\dots,y_{d-1}\}$ but not $D$. Such a set would be a member of $\mathcal{F}_x$ that does not intersect $D$, a contradiction.

\end{proof}






\section{Proofs}

\subsection{Proof of Theorem \ref{model1d}}

The dual of the $d$-cover-free property is that for every elements $x_1,\dots,x_{d+1}$ we cannot have that the sets that contain $x_{d+1}$ all contain at least one of the other $x_i$'s.
Let $$\cH_x:=\{F\setminus\{x\}: F\in \cF_x\},$$ and $\tau(\cH_x)$ be the size of the smallest set that intersects every member of $\cH_x$. With these notation the following lemma finishes the proof of Theorem \ref{model1d}.

\begin{lem} An element $x$ always finds out if he is defective if and only if $\tau(\cH_x)> d$.

\end{lem}
\begin{proof} If $x$ gets a NO answer, he learns he is not defective, thus we can assume he only gets YES answers.
If $\cH_x$ cannot be covered by at most $d$ elements different from $x$, then the only way to get YES answer to every element of $\cF_x$ is if $x$ is defective (as defective elements cover the sets that get YES answers). On the other hand if $\cH_x$ can be covered by at most $d$ elements different from $x$, then $x$ cannot exclude the possibility that those are the defective elements, together with arbitrary additional elements to reach $d$ defectives.

\end{proof}

\subsection{Proof of Theorem \ref{model2''d}}

\begin{lemma}\label{chard}

$\mathcal{F} \subset 2^{[n]}$ solves Model $2_d''$ if and only if the following two properties hold:

\vspace{2mm}

$\bullet_1$ for any two different $d$-element sets $X,Y \subset [n]$ there is $F \in \mathcal{F}$ with $F \cap X\neq \emptyset$ and 

\hspace{4mm}$F \cap Y = \emptyset$, and

\vspace{2mm}

$\bullet_2$ for any three different $d$-element sets $X,Y,Z \subset [n]$ there is $F \in \mathcal{F}$ with ($F \cap X \neq \emptyset$ 

\hspace{4mm} and $F \cap Y \neq \emptyset$ and $F \cap Z = \emptyset$) or ($F \cap X \neq \emptyset$ and $F \cap Z \neq \emptyset$ and $F \cap Y = \emptyset$).

\end{lemma}

\begin{proof}

1. Note that the property that Questioner can find out the answer is:
for any two different $d$-element sets $X,Y \subset [n]$ there is $F \in \mathcal{F}$ with ($F \cap X\neq \emptyset$ and $F \cap Y = \emptyset$) or ($F \cap Y \neq \emptyset$ and $F \cap X = \emptyset$). This property is contained in $\bullet_1$. 

\vspace{3mm}

Let us assume now $X$ is a set of size $d$.

\vspace{1mm}

2. If $X$ is the set of defectives, they have to find this out. It means that for a different $d$-element set $Y$, there should be an $F \in \mathcal{F}$ with $X \cap F \neq \emptyset$ and $Y \cap F = \emptyset$.

\vspace{2mm}

3. If $X$ is not the set of defectives, then another set $Y$ is, and they have to identify $Y$. Thus for a third $d$-element set $Z$, there should be a set that intersects $X$ (so they know the answer for it), and distinguishes $Y$ and $Z$, i.e. it intersects exactly one of them.

\end{proof}

\begin{lemma}\label{dual}

$\mathcal{F} \subset 2^{[n]}$ satisfies properties $\bullet_1$ and $\bullet_2$ if and only if its dual $\mathcal{G}$ is $d$-union-free and $\mathcal{G}^d$ is Sperner and intersection cancellative.

\end{lemma}

\begin{proof}

The dual of $\bullet_1$ is the following statement:

\vspace{2mm}

$\bullet_3$ for two different subfamilies each consisting of $d$ sets $\{F_1,...,F_d\}, \{G_1,...,G_d\} \subset \cF$ 

\hspace{4mm}there is $f \in [n]$ with $f \in \cup_{i=1}^d F_i \setminus \cup_{i=1}^d G_i$. 

\vspace{2mm}



The dual of $\bullet_2$ is the following statement:

\vspace{2mm}

$\bullet_4$ for three different subfamilies each consisting of $d$ sets $$\{F_1,...,F_d\}, \{G_1,...,G_d\},
\{H_1,...,H_d\} \subset \cF$$ 

\hspace{4mm} there is $f \in [n]$ with either $$f \in (\cup_{i=1}^d F_i \cap \cup_{i=1}^d G_i) \setminus \cup_{i=1}^d H_i, \textrm{ or }$$ $$f \in (\cup_{i=1}^d F_i \cap \cup_{i=1}^d H_i) \setminus \cup_{i=1}^d G_i.$$ 

\vspace{4mm}

It is easy to see that $\bullet_3$ is equivalent to the statement that $\cG$ is $d$-union-free and $\mathcal{G}^d$ is Sperner. Now we claim that $\bullet_4$ means that $\cG^d$ is intersection cancellative. Let us use the following notation:
$$F:=\cup_{i=1}^d F_i, \ \ G:= \cup_{i=1}^d G_i, \ \ H:= \cup_{i=1}^d H_i.$$

\noindent
Using these, the existence of $f$ means either $F \cap G \not \subset H$ or $F \cap H \not \subset G$. Let us define three properties.

\vspace{1mm}
$\circ_1$ $F \cap G \not \subset H$.

\vspace{1mm}

$\circ_2$ $F \cap H \not \subset G$.

\vspace{1mm}

$\circ_3$ $H \cap G \not \subset F$.

\vspace{1mm}

\noindent
Property $\bullet_2$ (for these three sets in this order) means that at least one of $\circ_1$ and $\circ_2$ holds. Considering the same three sets in different orders we get that also at least one of $\circ_1$ and $\circ_3$ and one of $\circ_3$ and $\circ_2$ holds. It is true if and only if at least two of these three properties hold.

\vspace{2mm}

\noindent
To finish the proof of Lemma \ref{dual} we prove the following:

\begin{clm}\label{canc}

A family $\mathcal{H} \subset 2^{[n]}$ is intersection cancellative if and only if at least two out of $\circ_1, \circ_2$ and $\circ_3$ hold for any three members of it. 

\end{clm}

\begin{proof}

Let us assume $\cF'$ is intersection cancellative and let $F,G,H \in \mathcal{H}$. Let us assume at most one, say $\circ_3$ of the three properties holds, thus $\circ_1$ and $\circ_2$ do not hold. The first one implies $F\cap G \subset H$, and obviously $F \cap G \subset F$. Thus we have $F \cap G \subset F \cap H$. Similarly the second one implies $F \cap H \subset F \cap G$, hence they together imply $F\cap H=F\cap G$, which contradicts the intersection cancellative property and our assumption that $F,G,H$ are three different sets.

Let us assume now that $\cH$ is not intersection cancellative, thus we have $F\cap G=F\cap H$. This implies both $F\cap G \subset H$ and $F\cap H\subset G$, thus at most one of  $\circ_1, \circ_2$ and $\circ_3$ can hold.

\end{proof}

We are done with the proof of Theorem \ref{model2''d}.

\end{proof}

\section{Adaptive scenario}

A natural idea is to consider the adaptive versions of these problems. Here we assume the Questioner knows all the earlier answers, and then he can choose the next query. He can find the defective, and then use further queries to share some information with the elements. However, there are two versions of this problem. The elements might know the algorithm, and use the order of the queries to gain information, or they only receive the answers to the queries at the end in no particular order. 

For example in Model $2_d$ in the second version we require that for every element $x$ the family $\cF_x$ with the answers is enough to find all the defectives, i.e. for two distinct sets $D,D'$ of size $d$ there is a query that contains $x$ and intersects only one of $D$ and $D'$. It is still obviously not solvable, as every defective element only gets YES answers and no information about the others. However, in the first version Questioner may start with a $d$-separating family, then ask the set of defectives and then the set of non-defectives. This way every element has to look only at the last query that contains it. If the answer to that is YES, then it is the set of defectives, if the answer is NO, it is the set of non-defectives. In both cases the defectives are identified.


From now on we consider only the second version, i.e. the elements receive the answer to the queries containing them at the end of the algorithm in no particular order, and they only know the underlying set and the number of defectives. It is still possible for the Questioner to find the defective, and then share some information using further queries.

Let $t^a(d,n)$ denote the number of queries in the fastest adaptive algorithm that finds the $d$ defective (we mentioned some inequalities on $t^a(d,n)$ in the introduction), then $t^a(d,n)$ is a lower bound in every model. On the other hand $t^a(d,n)+d+1$ queries are enough in Model $1_d$, $t^a(d,n)+1$ queries are enough in Model $2_d'$ and $t^a(d,n)+d+1$ queries are enough in Model $2_d''$: first Questioner finds the $d$ defectives, then ask them as singletons, and/or the set of non-defectives. 

Let us consider now Model $3_d$. By Corollary \ref{model3dcor} there is a solution for $n$ large enough, but that solution is linear in $n$. On the other hand it can be seen easily that for $n=d+1$ there is no solution even adaptively. Here we give a faster algorithm.

\begin{thm} There is an adaptive algorithm that solves Model $3_d$ and uses at most $2d\log_2 n+5d$ queries if $n$ is large enough.

\end{thm}

\begin{proof} Questioner starts with asking a query $Q$ of size $\lfloor n/2\rfloor$ and its complement. Then in the next round he asks two complementing subsets of size differing by at most one in every query that was answered YES (say $Q_1$ and $Q_2$ with $Q_1\cup Q_2=Q$). He repeats this in every round except if the subset has size at most 5, he stops and does not ask that subset as a query.  Since he asks disjoint sets in every round, he gets at most $d$ YES answers, thus there are at most $2d$ queries in the next round. There are obviously at most $\log_2 n$ rounds. After that we have a family $\cD$ of at most $d$ sets of size at most $5$, each containing at least one defective element. Let $A:=\{a_1,\dots, a_l\}$ be their union ($l\le 5d$), we also know that every defective is in $A$. Let $D_i\in \cD$ be the set that contains $a_i$. As $n$ is large enough, we can assume that there were two queries $B$ and $C$ that were answered NO and have size at least $5d$. Let $b_1,\dots, b_{5d}$ be distinct elements of $B$ and $c_1,\dots, c_{5d}$ be distinct elements of $C$. Then Questioner also asks the queries $\{a_i,b_i,c_i\}$ for $i\le l$.

As we know $b_i$ and $c_i$ are not defective, Questioner finds out if $a_i$ is defective for every $i$. On the other hand, if $a_i$ is defective, every query $Q$ that contains $a_i$ also contains either other elements of $D_i$ or contains $b_i$, thus $a_i$ cannot be sure he is defective. If $a_i$ is not defective, then all he knows is that another element of $D_i$ is defective, but there are more than one such elements. Any other element $x$ appears in a query $Q_1$ that got answer NO. At this point all they know is that there is another set $Q_2$ that contains a defective. $Q_2$ has size at least $3$, thus $x$ does not know at this point which one is defective. If $x=b_i$ or $x=c_i$ for some $i$, he can get additional information about only one element of $Q_2$, thus there are two candidates remaining. Finally, if the answer to $\{a_i,b_i,c_i\}$ is YES, then $x$ does not know if $a_i$ or $c_i$ is the defective.

\end{proof}



It is easy to see that Model $4_d$ still cannot be solved if $i\ge d$ or $j<d$. Indeed, the defectives still get only YES answers, thus less than $d$ of them cannot have any idea about the remaining defectives. On the other hand we will show that there are possible answers such that the defectives together will find out they are the defectives, showing $i\ge d$ is impossible. Let us assume that every answer is YES, unless it is impossible. If Questioner finds out that $D$ is the set of defectives, it means that for every other set $D'$ of size $d$ there was a query at some point that intersected exactly one of $D$ and $D'$. At that point YES was a possible answer, thus the answer to that query was YES. Hence it intersected $D$ and was disjoint from $D'$. Then an element of $D$ knows $D'$ is not the set of defectives, and this holds for every set $D'\neq D$ of size $d$.


\section{Remarks}

We finish this article with some possible directions that can be investigated:

\vspace{2mm}

$\bullet$ In some of the above models we proved that there is a family that solves the model, but did not say anything about its possible size.

\vspace{2mm}

$\bullet$ In case of Model $4_d$ our results can only be considered as the starting point of the investigations. In particular, it would be interesting to see if $i$ can go above $1$. It is tempting to try to extend the proof of Theorem \ref{model3d} to this case, and use a linear hypergraph of large girth. However, it does not work even for $i=2$. The property that the defectives can be identified forces the elements to be contained in many hyperedges, while the property that no 2 elements can identify any of the defectives forces the opposite. 

If the query hypergraph is linear and $d$-separating, it is easy to see that for two elements contained by the same query there must be at least $d$ other sets containing the two elements. This implies almost every element has to be contained in more than $(d+1)/2$ queries.

On the other hand let us consider two elements $x$, $y$ that are not contained in the same hyperedge. The large girth of the query hypergraph implies that there is at most one other element $z$ contained in a hyperedge $Q$ together with $x$ and another hyperedge together with $y$ (if there is no such $z$, then let $Q$ be an arbitrary query containing $x$). Let us assume the answer to $Q$ is NO, and the answer to every other query containing $x$ or $y$ is YES. Then $x$ and $y$ together know that $x$ is not defective. If they cannot identify $y$ as a defective, there cannot be more than $d$ hyperedges containing $x$ or $y$ besides $Q$. This implies almost every element has to be contained in at most $(d+1)/2$ queries.

\vspace{2mm}

$\bullet$ In ~\cite{GV2017} we considered the abstract version of the model introduced by Tapolcai et al.~\cite{TRHHS2014,TRHGHS2016}. Here we extended our models to the case of more defectives. It would be interesting to see if their model can be extended similarly.

\vspace{2mm}

$\bullet$ It is a phenomenon in combinatorial group testing that in most of the models the adaptive version actually means two round version of the problem (see e.g.~\cite{DGV2005}) Recently there was some interest in the $r$ round (or multi-stage) versions of combinatorial group testing problems, where this phenomenon does not hold (see e.g.~\cite{DMT2013,GV2016}). It would be interesting to investigate these models in this context. 

\vspace{2mm}

$\bullet$ One can consider a variant of these models, where instead of requiring that the elements find all (or none) of the defective elements, we require that they identify at least $i$ and/or at most $j$ of them.

\subsection*{Acknowledgement}

We would also like to thank all participants of the Combinatorial Search Seminar at the Alfr\'ed R\'enyi Institute of Mathematics for fruitful discussions.

\end{document}